\documentclass[12pt, a4paper, reqno]{amsart}
\usepackage{mathrsfs}
\usepackage{bbm}
\usepackage{amsmath,amscd,amssymb,latexsym}
\usepackage[all]{xy}

\input xypic

\evensidemargin=\oddsidemargin
\DeclareMathVersion{can}
\DeclareMathAlphabet{\can}{OT1}{cmss}{m}{n}
\newtheorem{thm}{Theorem}[section]

\newtheorem{lem}[thm]{Lemma}

\newtheorem{exa}[thm]{Example}
\theoremstyle{definition}

\theoremstyle{fact}

\theoremstyle{conjecture}

\numberwithin{equation}{section}

\newcommand{\ord}{\operatorname{ord}}

\begin{document}
\title[Cyclic codes]{Weight distribution of two classes of cyclic codes with  respect to two distinct order elements}

\author[C. Li] {Chengju Li}
\address{\rm Department of Mathematics, Nanjing University of Aeronautics and Astronautics,
Nanjing, 211100, P.R. China
} \email{lichengju1987@163.com}
\author[Q. Yue]{Qin Yue}
\address{\rm Department of Mathematics, Nanjing University of Aeronautics and Astronautics,
Nanjing, 211100, P. R. China} \email{yueqin@nuaa.edu.cn}
\thanks{The paper is supported by NNSF of China (No. 11171150)}

\subjclass[2000]{94B15, 11T71, 11T24}
 \keywords{Weight distribution; Cyclic codes; Gauss periods; Character sums}
\begin{abstract} Cyclic codes are an interesting type of linear codes and have wide
applications in communication and storage systems due to their efficient encoding and
decoding algorithms.  Cyclic codes  have been studied
for many years, but their weight distribution are known only for a few cases.
In this paper, let $\Bbb F_r$ be an extension of a finite field $\Bbb F_q$ and  $r=q^m$,
we determine the weight distribution of the cyclic codes
$\mathcal C=\{ c(a, b): a, b \in \Bbb F_r\},$
$$c(a, b)=(\mbox {Tr}_{r/q}(ag_1^0+bg_2^0),  \ldots, \mbox {Tr}_{r/q}(ag_1^{n-1}+bg_2^{n-1})), g_1, g_2\in \Bbb F_r,$$
in the following two cases: (1) $\ord(g_1)=n, n|r-1$ and $g_2=1$; (2) $\ord(g_1)=n$,
$g_2=g_1^2$,  $\ord(g_2)=\frac n 2$, $m=2$ and $\frac{2(r-1)}n|(q+1)$.
\end{abstract}
\maketitle

\section{Introduction}

Let $\Bbb F_q$ be a finite field with $q$ elements, where $q=p^s$, $p$ is a prime, and
$s$ is a positive integer. An $[n, k, d]$ linear code $\mathcal{C}$ is a $k$-dimensional
subspace of $\Bbb F_q^n$ with minimum distance $d$. It is called cyclic if
$(c_0, c_1, \ldots, c_{n-1}) \in \mathcal C$ implies
$(c_{n-1}, c_0, c_1, \ldots, c_{n-2}) \in \mathcal C$.
By identifying the vector $(c_0, c_1, \ldots, c_{n-1}) \in \Bbb F_q^n$ with
$$c_0+c_1x+c_2x^2+\cdots+c_{n-1}c^{n-1} \in \Bbb F_q[x]/(x^n-1),$$
any code $\mathcal C$ of length $n$ over $\Bbb F_q$ corresponds to a subset of
$\Bbb F_q[x]/(x^n-1)$. Then $\mathcal C$ is a cyclic code if and only if the
corresponding subset is an ideal of $\Bbb F_q[x]/(x^n-1)$. Note that every
ideal of $\Bbb F_q[x]/(x^n-1)$ is principal. Hence there is a monic polynomial $g(x)$ of the least
degree such that $\mathcal C=\langle g(x) \rangle$ and $g(x) \mid (x^n-1)$. Then $g(x)$ is called the generator
polynomial and $h(x)=(x^n-1)/g(x)$ is called the parity-check polynomial of the cyclic code
$\mathcal C$. Suppose that $h(x)$ has $t$ irreducible factors over $\Bbb F_q$, we call
$\mathcal C$ the dual of the cyclic code with $t$ zeros.

 Let $A_i$ be the number of codewords
with Hamming weight $i$ in the code $\mathcal C$ of length $n$. The weight enumerator of
$\mathcal C$ is defined by
$$1+A_1x+A_2x^2+\cdots+A_nx^n.$$
The sequence $(1, A_1, A_2, \ldots, A_n)$ is called the weight distribution of the code
$\mathcal C$.  In coding theory it is often desirable to know the weight distributions
of the codes because they can be used to estimate the error correcting capability and
the error probability of error detection and correction with respect to some algorithms.
This is quite useful in practice. Unfortunately, it is a very hard problem in general
and remains open for most cyclic codes.

 Let $r=q^m$ for a positive integer $m$ and $\alpha$ a generator of $\Bbb F_r^\ast$.
Let $h(x)=h_1(x)h_2(x) \cdots h_t(x)$, where $h_i(x) (1 \leq i \leq t)$ are distinct monic irreducible
polynomials over $\Bbb F_q$.  Let $g_i^{-1}=\alpha^{-s_i}$  be a root of $h_i(x)$ and $n_i$ the order of $g_i$
for $0 \leq s_i \leq r-2(1 \leq i \leq t)$. Denote $\delta=\gcd(r-1, s_1, s_2, \ldots, s_t)$ and $n=\frac {r-1} \delta$.
The cyclic code $\mathcal C$ can be defined by
$$\mathcal C=\{ c(a_1, a_2, \ldots, a_t) : a_1, a_2, \ldots, a_t \in \Bbb F_r\},$$
where $$ c(a_1, a_2, \ldots, a_t)=(\mbox{Tr}_{r/q}(\sum_{i=1}^{t}a_ig_i^0), \mbox{Tr}_{r/q}(\sum_{i=1}^{t}a_ig_i^1), \ldots, \mbox{Tr}_{r/q}(\sum_{i=1}^{t}a_ig_i^{n-1}))$$ and
$\mbox{Tr}_{r/q}$ denotes the trace function from $\Bbb F_r$ to $\Bbb F_q$.
It follows from Delsartes Theorem \cite{D} that the code $\mathcal C$ is an $[n, k]$ cyclic code
over $\Bbb F_q$ with the parity-check polynomial $h(x)$, where $k=\deg (h_1(x))+\deg (h_2(x))+ \cdots +
\deg (h_t(x))$.  The weight distributions of such cyclic codes have been studied for many years and are known
in some cases. We describe the known results as follows.
 \begin{enumerate}
   \item For $t=1$, $\mathcal C$ is called an irreducible cyclic code. The weight distributions
   of irreducible cyclic codes have been extensively studied and can be found in
   \cite{BM72, BMY, DG, D09, DY, MC1, MC2}.
   \item  For $t=2$, i.e., $h(x)=h_1(x)h_2(x)$. The duals of the cyclic codes with two zeros have been
   well investigated when $\deg (h_1(x))=\deg (h_2(x))$.  If $g_1$ and
   $g_2$ have the same order in $\Bbb F_r^\ast$, we know $\deg (h_1(x))=\deg (h_2(x))$. Then the weight distribution of such
   cyclic codes had been determined for some special cases \cite{DMLZ, FL, FM, Ma, V, W, X1, X2, X3, ZD}.
   If $\Bbb F_r^\ast=\langle g_1 \rangle=\langle g_2 \rangle$, then the weight distribution of
   the code $\mathcal C$ which is called the dual of primitive cyclic code with two zeros had been studied
   in \cite{BM, CCD, CCZ, C, F, LTW, Mc, M, S, YCD}.
   \item For $t=3$.  The results on the weight distribution of  cyclic codes $\mathcal C$  with three zeros can be found
   in \cite{LF, Z, ZDLZ}.
   \item For arbitrary $t$. Yang et al. \cite{YXD} described a class of the duals of cyclic codes with $t$ zeros and determined their weight
   distributions under special conditions. Li et al. \cite{LHFG} also studied such cyclic codes  and developed a connection
   between the weight distribution and the spectra of Hermitian forms graphs.
 \end{enumerate}

In this paper, we shall determine the weight distributions
 of two classes of cyclic codes whose duals have two zeros.
 Let $\alpha$ be a primitive element of $\Bbb F_r$ and $r-1=nN$ for two positive integers $n>1$ and $N>1$.
 We mainly consider the following two cases of the cyclic code $\mathcal C$.

 \begin{enumerate}
   \item  Assume that the order of $g_1$ is $n$ and the order of $g_2$  is $1$. Then we can set $g=g_1=\alpha^N$, $g_2=1$, and
   $$\mathcal C_1=\{ c(a, b): a, b \in \Bbb F_r\},$$
   where $$c(a, b)=(\mbox {Tr}_{r/q}(ag^0+b), \mbox {Tr}_{r/q}(ag^1+b), \ldots, \mbox {Tr}_{r/q}(ag^{n-1}+b)).$$
   It is obvious that the order of $g$ is not equal to $1$ and thus the parity-check polynomial of $\mathcal C_1$ is
   $(x-1)h_{g^{-1}}(x)$, where $h_{g^{-1}}(x)$ is the minimal polynomial of $g^{-1}$ over $\Bbb F_q$. The lower bound on the minimum weight of
$\mathcal C_1$ had been given by Ding \cite{D13}. We can also get a tight bound on the minimum weight of such cyclic code.
   \item Assume that the order of $g_1$ is $n$($n$ is even) and the order of $g_2$  is $\frac n 2$. Then we can set $g=g_1=\alpha^N$
   and $g_2=\mu g^2$ for $\mu \in \Bbb F_r^\ast$, where the order of $\mu$ is a divisor of $\frac n 2$. Denote
      $$\mathcal C_2=\{ c(a, b): a, b \in \Bbb F_r\},$$
      where $$c(a,b)=(\mbox {Tr}_{r/q}(ag^0+b(\mu g^2)^0), \mbox {Tr}_{r/q}(ag^1+b(\mu g^2)^1), \ldots, \mbox {Tr}_{r/q}(ag^{n-1}+b(\mu g^2)^{n-1})).$$
   It is easily known that $g$ and $\mu g^2$ are not conjugates of each other due to their distinct orders
    and then $h_{g^{-1}}(x)$ and $h_{(\mu g^2)^{-1}}(x)$ are distinct, where
   $h_{g^{-1}}(x)$ and $h_{(\mu g^2)^{-1}}(x)$ are the minimal polynomial of $g^{-1}$ and $(\mu g^2)^{-1}$ over $\Bbb F_q$, respectively.
   We know that the parity-check polynomial of $\mathcal C_2$ is $h_{g^{-1}}(x)h_{(\mu g^2)^{-1}}(x)$.
   For convenience, we shall present a method to determine the
   weight distribution of the cyclic code $\mathcal C_2$ when $\mu=1$. The general case can be also dealt with
   by our method and the method in \cite{DMLZ} or \cite{Ma}. In addition, we present a tight bound on the minimum weight of cyclic code
   $\mathcal C_2$.
 \end{enumerate}

 This paper is organized as follows. In Section 2, we introduce some results about Gauss periods.
 In Section 3 and 4, we shall determine the weight distributions of the
  cyclic codes $\mathcal C_1$ and $\mathcal C_2$, respectively.

\section{Gauss periods}

Let ${\Bbb F}_r$ be the finite field with $r$ elements, where $r$ is a power of prime $p$.
For any $a \in {\Bbb F}_r$, we can define
an additive character of the finite field ${\Bbb F}_r$ as follows:
$$\psi_a : {\Bbb F}_r \rightarrow {\Bbb C}^\ast, \psi_a(x)=\zeta_p^{\mbox{Tr}_{r/p}(ax)},$$
where $\zeta_p=e^{\frac{2\pi i}p}$ is a $p$-th primitive  root of unity and $\mbox{Tr}_{r/p}$ denotes
the trace from ${\Bbb F}_r$ to ${\Bbb F}_p$. If $a=1$, then $\psi_1$ is
called the canonical additive character of $\Bbb F_r$. The orthogonal property of additive characters which can
be found in \cite{LN} is given by
$$\sum_{x \in \Bbb F_r} \psi_1(ax)=\left\{\begin{array}{ll}
r, \mbox{ if } a=0; \\
0, \mbox{ if } a \in {\Bbb F}_r ^\ast. \end{array}\right.$$

Let $r-1=nN$ and $\alpha$ a fixed primitive element of $\Bbb F_r$, where $r=q^m=p^{sm}$.
We define $C_i^{(N, r)}=\alpha^i \langle \alpha^N \rangle$ for $i=0, 1, \ldots, N-1$,
where $\langle \alpha^N \rangle$ denotes the subgroup of $\Bbb F_r^\ast$ generated
by $\alpha^N$. The Gauss periods of order $N$ are given by
$$\eta_i^{(N, r)}=\sum_{x \in C_i^{(N, r)}}\psi(x),$$
where $\psi$ is the canonical additive character of $\Bbb F_r$ and $\eta_i^{(N, r)}=\eta_{i \pmod N}^{(N, r)}$ if $i \geq N$.
In general, the explicit evaluation of Gauss periods is a very difficult problem.
 However, they can be computed in a few cases.

\begin{lem} \cite{My} When $N=2$, the Gauss periods are given by
$$\eta_0^{(2, r)}=\left\{\begin{array}{ll}
\frac {-1+(-1)^{sm-1} \sqrt r} 2, \mbox{ if } p \equiv 1 \pmod 4, \\
\frac {-1+(-1)^{sm-1} (\sqrt {-1})^{sm} \sqrt r} 2, \mbox{ if }  p \equiv 3 \pmod 4, \end{array}\right.$$
and $\eta_1^{(2, r)}=-1-\eta_0^{(2, r)}$.
\end{lem}

\begin{lem} \cite{My} Let $N=3$. If $p \equiv 1 \pmod 3$ and $sm \equiv 0 \pmod 3$, then
$$\eta_0^{(3, r)}=\frac {-1+c_1 r^{\frac 1 3}} 3 $$ and
$$\{\eta_1^{(3, r)}, \eta_2^{(3, r)}\}=\{
\frac {-1-{\frac 1 2}(c_1+9d_1) r^{\frac 1 3}} 3,
\frac {-1-{\frac 1 2}(c_1-9d_1) r^{\frac 1 3}} 3\},$$
where $c_1$ and $d_1$ are given by $4p^{\frac {sm} 3}=c_1^2+27d_1^2, c_1 \equiv 1 \pmod 3$, and
$\gcd(c_1, p)=1$.
\end{lem}

\begin{lem} \cite{My} Let $N=4$. If $p \equiv 1 \pmod 4$ and $sm \equiv 0 \pmod 4$, then
$$\eta_0^{(4, r)}=\frac {-1-r^{\frac 1 2}-2s_1r^{\frac 1 4}} 4,
\eta_2^{(4, r)}=\frac {-1-r^{\frac 1 2}+2s_1r^{\frac 1 4}} 4 $$ and
$$\{\eta_1^{(4, r)}, \eta_3^{(4, r)}\}=\{
\frac {-1+r^{\frac 1 2}+4t_1r^{\frac 1 4}} 4,
\frac {-1+r^{\frac 1 2}-4t_1r^{\frac 1 4}} 4\},$$
where $s_1$ and $t_1$ are given by $p^{\frac {sm} 2}=s_1^2+4t_1^2, s_1 \equiv 1 \pmod 4$, and
$\gcd(s_1, p)=1$.
\end{lem}

The Gauss periods in the semi-primitive case are known and are described in the following lemma.

\begin{lem} \cite{My} Assume that there exists a least positive integer $e$ such that
$p^e \equiv -1\pmod N$. Let $r=p^{2ef}$ for some positive integer $f$.
\begin{enumerate}
  \item If $f, p$, and $\frac {p^e+1} N$ are all odd, then
  $$\eta_{N/2}^{(N, r)}=\frac {(N-1)\sqrt r-1} N, \eta_i^{(N, r)}=-\frac {\sqrt r+1} N \mbox { for }
  i \neq N/2.$$
  \item In all other cases,
  $$\eta_0^{(N, r)}=\frac {(-1)^{f+1}(N-1)\sqrt r-1} N, \eta_i^{(N, r)}=\frac {(-1)^f\sqrt r-1} N
  \mbox { for } i \neq 0.$$
\end{enumerate}
\end{lem}

The Gauss periods in the index $2$ case can be described in the following lemma.

\begin{lem} \cite{DY, My} Let $N > 3$  be a prime with $N \equiv 3 \pmod 4$, $p$ a prime such
that $[\Bbb Z_N^\ast : \langle p \rangle]=2$, and $r=p^{\frac {N-1} 2 k}$ for some positive integer $k$.
Let $h$ be the class number of $\Bbb Q(\sqrt {-N})$ and $a,b$ the integers satisfying
$$\left\{\begin{array}{lll}
4p^h=a^2+Nb^2 \\
a \equiv -2p^{\frac {N-1+2h}4} \pmod N \\
b>0, p \nmid b. \end{array}\right.$$
Then the Gauss periods of order $N$ are given by
$$\left\{\begin{array}{lll}
\eta_0^{(N, r)}=\frac 1 N (P^{(k)}A^{(k)}(N-1)-1) \\
\eta_u^{(N,r)}=\frac {-1} N(P^{(k)}A^{(k)}+P^{(k)}B^{(k)}N+1), \mbox{ if } (\frac u N)=1 \\
 \eta_u^{(N,r)}=\frac {-1} N(P^{(k)}A^{(k)}-P^{(k)}B^{(k)}N+1), \mbox{ if } (\frac u N)=-1, \end{array}\right.$$
 where
 $$\left\{\begin{array}{lll}
P^{(k)}=(-1)^{k-1}p^{\frac k 4(N-1-2h)} \\
A^{(k)}=\mbox{Re}(\frac {a+b\sqrt{-N}} 2)^k \\
B^{(k)}=\mbox{Im}(\frac {a+b\sqrt{-N}} 2)^k/ \sqrt N. \end{array}\right.$$
\end{lem}

In the following lemma, we introduce a bound on the values of Gauss periods which can be found in \cite{DY}.

\begin{lem} \cite{DY}  For all $i$ with $0 \leq i \leq N-1$, we have
$$|\eta_i^{(N, r)}+ \frac 1 N| \leq \frac {(N-1) \sqrt r} N.$$
\end{lem}

\section{ The weight distribution of $\mathcal C_1$}

Let $\alpha$ be a primitive element of $\Bbb F_r$, where $r=q^m$. Let $r-1=nN$ for two positive integers $n>1$ and $N>1$.
For $g=\alpha^N$ we define a cyclic code over $\Bbb F_q$ by
$$\mathcal C_1=\{ c(a, b): a, b \in \Bbb F_r\},$$
   where $$c(a, b)=(\mbox {Tr}_{r/q}(ag^0+b), \mbox {Tr}_{r/q}(ag^1+b), \ldots, \mbox {Tr}_{r/q}(ag^{n-1}+b)).$$
Then $\mathcal C_1$ is an $[n, k+1]$ cyclic code with parity-check polynomial
$(x-1)h_{g^{-1}}(x)$, where $k$ is the order of $q$ modulo $n$ and $h_{g^{-1}}(x)$ is the
minimal polynomial of $g^{-1}$ over $\Bbb F_q$. The lower bound on the minimum weight of
$\mathcal C_1$ had been given by Ding \cite{D13}. For any $a, b \in \Bbb F_r$, the Hamming weight
of $c(a, b)$ is equal to $$W_H(c(a, b))= n-Z(r, a, b),$$ where
$$Z(r, a, b)=|\{x \in C_0^{(N, r)} : \mbox{Tr}_{r/q}(ax+b)=0\}|.$$

Let $\phi$ be the canonical additive character of $\Bbb F_q$. Then $\psi=\phi \circ \mbox{Tr}_{r/q}$
is the canonical additive character of $\Bbb F_r$. By the orthogonal property of additive characters we have
\begin{eqnarray*} Z(r,a,b)&=&\sum_{x \in C_0^{(N, r)}}\frac 1 q \sum_{y \in \Bbb F_q}\phi(y \mbox{Tr}_{r/q}(ax+b)) \\
&=& \sum_{x \in C_0^{(N, r)}}\frac 1 q \sum_{y \in \Bbb F_q}\phi( \mbox{Tr}_{r/q}(yax+yb)) \\
&=& \frac 1 q \cdot \frac{r-1} N+\frac 1 q \sum_{x \in C_0^{(N, r)}} \sum_{y \in \Bbb F_q^\ast}\psi(yax+yb)\\
&=& \frac{r-1} {qN}+\frac 1 q \sum_{y \in \Bbb F_q^\ast}\psi(yb)\sum_{x \in C_0^{(N, r)}} \psi(yax).
\end{eqnarray*}

To compute the values of $Z(r, a, b)$, we introduce the following lemma which can be found in \cite{J}.
\begin{lem} Let $H$ and $K$ be two subgroups of a finite Abelian group $G$.
Then we have
$$h_1K=h_2K \mbox{ if and only if }  h_1(H \cap K)=h_2(H \cap K)$$
for $h_1, h_2 \in H$. Moreover,
$$HK/K \cong H/(H \cap K) \mbox{ and } [HK : K]=[H :(H \cap K)],$$
where $HK=\{hk : h \in H, k \in K\}.$
\end{lem}

\begin{thm} Let $d=\gcd(N,\frac {r-1} {q-1})$. Note that $c(a, b_1)=c(a, b_2)$ if $\mbox{Tr}_{r/q}(b_1)=\mbox{Tr}_{r/q}(b_2)$.
 Then the weight distribution of $\mathcal C_1$ is given by Table 1.
In especial, if $N \mid \frac {r-1} {q-1}$, then the weight distribution of $\mathcal C_1$ is given by Table 2.
\end{thm}

\[ \begin{tabular} {c} Table 1. Weight distribution of $\mathcal C_1.$\\
\begin{tabular}{|c|c|}
  \hline
 Weight & Frequency ($\begin{array}{c}0 \leq l \leq d-1 \\ 0 \leq j \leq N-1 \\  0 \leq k \leq \frac N d-1\end{array}$) \\
  \hline
        0   &    1 \\
  $\frac {r-1} N$ & $q-1$ \\
  $\frac {(q-1)(r-1)} {qN}-\frac {(q-1)d} {N q}
  \eta_l^{(d,r)}$ & $\frac {r-1} {d}$ \\
  $\frac {(q-1)(r-1)} {qN}-\frac 1 q\sum\limits_{i=0}^{\frac N d-1} \eta_{\frac {r-1} {q-1}i+j}^{(N,r)}
   \eta_{i+k}^{(\frac N d, q)}$  & $\frac {d(q-1)(r-1)} {N^2}$ \\
  \hline
\end{tabular}
\end{tabular}
\]

\[ \begin{tabular} {c} Table 2. Weight distribution of $\mathcal C_1$ when $N \mid \frac {r-1} {q-1}$.\\
\begin{tabular}{|c|c|}
  \hline
 Weight & Frequency ($0 \leq j \leq N-1$)\\
  \hline
        0   &    1\\
  $\frac {r-1} N$ & $q-1$ \\
  $\frac {(q-1)(r-1)} {qN}-\frac {q-1}  q \eta_j^{(N, r)}$ & $\frac {r-1} {N}$ \\
  $\frac {(q-1)(r-1)} {qN}+\frac 1 q \eta_j^{(N, r)}$  & $\frac {(q-1)(r-1)} {N}$ \\
  \hline
\end{tabular}
\end{tabular}
\]

\begin{proof} Since $\alpha$ is a primitive element of $\Bbb F_r$, we have $\Bbb F_q^\ast=
\langle \alpha^{\frac {r-1} {q-1}} \rangle$.
Denote $H=\Bbb F_q^\ast$ and $K=C_0^{(N, r)}=\langle \alpha^N \rangle$. Then
$$HK=C_0^{(d, r)}=\langle \alpha^d \rangle, H \cap K=\langle \alpha^{\frac {r-1} {q-1} \cdot \frac N d } \rangle,$$
where $d=\gcd(N,\frac {r-1} {q-1})$. By Lemma 3.1 we have
$[HK : K]=\frac N d \mid (q-1)$ and
$$\Bbb F_q^\ast=\cup_{i=0}^{\frac N d-1}C_i^{(\frac N d, q)},$$ where
$$C_i^{(\frac N d, q)}=\alpha^{\frac {r-1} {q-1}i}\langle \alpha^{\frac {r-1} {q-1} \cdot \frac N d } \rangle
=\alpha^{\frac {r-1} {q-1}i}(H \cap K)
\mbox{ for } 0 \leq i \leq \frac N d-1.$$
Then by Lemma 3.1 again we have
$$HK=\Bbb F_q^\ast \cdot C_0^{(N, r)}=\cup_{i=0}^{\frac N d-1}\alpha^{\frac {r-1} {q-1}i}C_0^{(N, r)}
=\cup_{i=0}^{\frac N d-1}C_{\frac {r-1} {q-1}i}^{(N, r)}.$$ where
$C_{\frac {r-1} {q-1}i}^{(N, r)}=C_{\frac {r-1} {q-1}i \pmod N}^{(N, r)}$ . Hence
\begin{eqnarray*}
Z(r,a,b)&=& \frac{r-1} {qN}+\frac 1 q \sum_{y \in \Bbb F_q^\ast}\psi(yb)\sum_{x \in C_0^{(N, r)}} \psi(yax) \\
&=& \frac{r-1} {qN}+\frac 1 q \sum_{i=0}^{\frac N d-1}\sum_{y \in C_i^{(\frac N d, q)}}\psi(yb)
\sum_{x \in C_0^{(N, r)}} \psi(yax) \\
&=& \frac{r-1} {qN}+\frac 1 q \sum_{i=0}^{\frac N d-1}\sum_{y \in C_i^{(\frac N d, q)}}\psi(yb)
\sum_{x \in C_{\frac {r-1} {q-1}i}^{(N, r)}} \psi(ax) \\
&=& \frac{r-1} {qN}+\frac 1 q \sum_{i=0}^{\frac N d-1}\sum_{x \in C_{\frac {r-1} {q-1}i}^{(N, r)}} \psi(ax)
\sum_{y \in C_i^{(\frac N d, q)}}\psi(yb).
\end{eqnarray*}

Using the orthogonal property of additive characters again, we have
\begin{eqnarray*}
\sum_{y \in \Bbb F_q^\ast}\psi(yb)&=& \sum_{y \in \Bbb F_q^\ast}\phi(\mbox{Tr}_{r/q}(yb)) \\
  &=& \sum_{y \in \Bbb F_q^\ast}\phi(y \mbox{Tr}_{r/q}(b)) \\
  &=& \left\{\begin{array}{ll}
 q-1 , \mbox{ if } \mbox{Tr}_{r/q}(b)=0, \\
-1, \mbox{ if }  \mbox{Tr}_{r/q}(b)\neq 0. \end{array}\right.
\end{eqnarray*}

It is known that $\mbox{Tr}_{r/q}$ maps $\Bbb F_r$ onto $\Bbb F_q$ and
$$|\{ b \in \Bbb F_r : \mbox{Tr}_{r/q}(b)=c\}|=\frac r q$$ for each $c \in \Bbb F_q$.
Note that $c(a, b_1)=c(a, b_2)$ if $\mbox{Tr}_{r/q}(b_1)=\mbox{Tr}_{r/q}(b_2)$. Then we can
determine the values of $Z(r, a, b)$ and their frequencies as follows.

 \begin{enumerate}
   \item If $a=0, \mbox{Tr}_{r/q}(b)=0$. We have
   $$Z(r,a,b)=\frac 1 q \cdot \frac{r-1} N+\frac 1 q \cdot (q-1) \cdot \frac{r-1} N=\frac{r-1} N.$$
   This value occurs once.
   \item If $a=0, \mbox{Tr}_{r/q}(b) \neq 0$. We have
   $$Z(r,a,b)=\frac 1 q \cdot \frac{r-1} N+\frac 1 q \cdot (-1) \cdot \frac{r-1} N=0.$$
  This value occurs $q-1$ times.
   \item If $a \in C_l^{(d,r)}, \mbox{Tr}_{r/q}(b)=0$. Note that $HK=C_0^{(d, r)}=\cup_{i=0}^{\frac N d-1}C_{\frac {r-1} {q-1}i}^{(N, r)}.$ We have
   \begin{eqnarray*}Z(r,a,b)&=&\frac 1 q \cdot \frac{r-1} N+\frac 1 q
   \sum_{i=0}^{\frac N d-1} \frac {(q-1)d} N \sum_{x \in C_{\frac {r-1} {q-1}i}^{(N, r)}} \psi(ax) \\
   &=& \frac{r-1} {q N}+\frac {(q-1)d} {q N} \sum_{x \in C_0^{(d, r)}} \psi(ax) \\
   &=& \frac{r-1} {q N}+\frac {(q-1)d} {q N} \eta_l^{(d, r)}.
   \end{eqnarray*}
   This value occurs $\frac {r-1} N$ times.
   \item If $a \in C_j^{(N,r)}, \mbox{Tr}_{r/q}(b) \in C_k^{(\frac N d,q)}$. We have
 $$Z(r,a,b)=\frac 1 q \cdot \frac{r-1} N+\frac 1 q
   \sum_{i=0}^{\frac N d-1} \eta_{\frac {r-1} {q-1}i+j}^{(N,r)} \eta_{i+k}^{(\frac N d, q)}.$$
   This value occurs $\frac {r-1} N \cdot \frac {q-1} {\frac N d}=\frac {d(q-1)(r-1)} {N^2}$ times.
 \end{enumerate}

 Note that $W_H(c(a, b))= n-Z(r, a, b)$. Then Table 1 can be obtained and Table 2 follows from Table 1. This completes the proof.
\end{proof}

We have determined the weight distribution of the cyclic code $\mathcal C_1$ when Gauss periods of order $N$ are known.
Then we give the following examples.

\begin{exa} The case $N=2$.
   \begin{enumerate}
     \item Let $q=3$ and $r=27$. Then $d=\gcd(N,\frac {r-1} {q-1})=1$.
  By Lemma 2.1 we have
  $$\eta_0^{(2, 27)}=\frac {-1-3 \sqrt {-3}} 2, \ \ \eta_1^{(2, 27)}=\frac {-1+3 \sqrt {-3}} 2$$
  and $$\eta_0^{(2, 3)}=\frac {-1+ \sqrt {-3}} 2, \ \ \eta_1^{(2, 3)}=\frac {-1-\sqrt {-3}} 2.$$
  Then by Table 1 we know that the code $\mathcal C_1$ is a $[13, 4, 7]$ cyclic code
   over $\Bbb F_3$ with the weight enumerator
    $$1+26x^7+26x^9+26x^{10}+2x^{13}.$$
     \item  Let $q=5$ and $r=25$. Then $N \mid \frac {r-1} {q-1}$.
  By Lemma 2.1 we have
  $$\eta_0^{(2, 25)}=-3, \ \  \eta_1^{(2, 25)}=2.$$
  Then by Table 2 we know that the code $\mathcal C_1$ is a $[12, 3, 8]$ cyclic code
   over $\Bbb F_5$ with the weight enumerator
    $$1+12x^8+48x^9+48x^{10}+16x^{12}.$$
   \end{enumerate}
 \end{exa}

\begin{exa} For the case $N=3$, we have $N \mid \frac {r-1}{q-1}$ under the conditions of Lemma 2.2.
     Let $q=7$ and $r=7^3$.  By Lemma 2.2 we have
     $$ \eta_0^{(3, 7^3)}=2, \  \eta_1^{(3, 7^3)}=-12, \ \eta_2^{(3, 7^3)}=9.$$
     Then by Table 2 we know that the code $\mathcal C_1$ is a $[114, 4, 90]$ cyclic code
   over $\Bbb F_7$ with the weight enumerator
    $$1+114x^{90}+798x^{96}+684x^{98}+684x^{99}+114x^{108}+6x^{114}.$$
 \end{exa}

\begin{exa} For the case $N=4$, we have $N \mid \frac {r-1}{q-1}$ under the conditions of Lemma 2.3.
     Let $q=5$ and $r=5^4$.  By Lemma 2.3 we have
     $$\eta_0^{(4, 5^4)}=1, \eta_2^{(4, 5^4)}=-14, \mbox{ and }
     \{\eta_1^{(4, 5^4)}, \eta_3^{(4, 5^4)}\}=\{1, 11\}.$$
     Then by Table 2 we know that the code $\mathcal C_1$ is a $[156, 5, 116]$ cyclic code
   over $\Bbb F_5$ with the weight enumerator
    $$1+156x^{116}+624x^{112}+312x^{124}+1248x^{125}+624x^{127}+156x^{136}+4x^{156}.$$
 \end{exa}

\begin{exa} For the semi-primitive case. Let $q=5$ and $N=3$. Then there exists the least positive integer $e=1$ such that
$ 5 \equiv -1 \pmod 3$. Let $r=5^2$. By Lemma 2.4 we have
$$\eta_0^{(3, 5^2)}=3, \ \eta_1^{(3, 5^2)}=\eta_2^{(3, 5^2)}=-2.$$
Then by Table 2 we know that the code $\mathcal C_1$ is an $[8, 3, 4]$ cyclic code
   over $\Bbb F_5$ with the weight enumerator
    $$1+8x^4+64x^6+32x^7+20x^8.$$
 \end{exa}

\begin{exa} For the index $2$ case. Let $q=2$, $N=7$, and $r=2^6$. By Lemma 2.5 we have
$$\eta_0^{(7, 2^6)}=5, \ \eta_1^{(7, 2^6)}=\eta_2^{(7, 2^6)}=\eta_4^{(7, 2^6)}=-3,$$ and
$$\eta_3^{(7, 2^6)}=\eta_5^{(7, 2^6)}=\eta_6^{(7, 2^6)}=1.$$
Then by Table 2 we know that the code $\mathcal C_1$ is a $[9, 7, 2]$ cyclic code
   over $\Bbb F_2$ with the weight enumerator
    $$1+9x^2+27x^3+27x^4+27x^5+27x^6+9x^7+x^9.$$
 \end{exa}

We know that the explicit values of the Gauss periods are very hard to determine. Then we have the following tight bound on
the minimum weight of $\mathcal C_1$ which is denoted by $W_H(\mathcal C_1)$.

\begin{thm}
If $N \mid \frac {r-1} {q-1}$, $q \geq 3$, and $N < \sqrt r$, then we have
$$W_H(\mathcal C_1) \geq \frac {(q-1)(r-(N-1)\sqrt r)} {qN}.$$
\end{thm}

\begin{proof}
It is immediate from Lemma 2.6 and Theorem 3.2.
\end{proof}

 \section{ The weight distribution of $\mathcal C_2$}

 Let $\alpha$ be a primitive element of $\Bbb F_r$ and $r-1=nN$ for two positive integers $n>1$ and $N>1$.
   For $g=\alpha^N$ we define a cyclic code over $\Bbb F_q$ by
  $$\mathcal C_2=\{ c(a, b): a, b \in \Bbb F_r\},$$
      where $$c(a,b)=(\mbox {Tr}_{r/q}(ag^0+b(g^2)^0), \mbox {Tr}_{r/q}(ag^1+b(g^2)^1), \ldots, \mbox {Tr}_{r/q}(ag^{n-1}+b(g^2)^{n-1})).$$
   It is known that the parity-check polynomial of $\mathcal C_2$ is $h_{g^{-1}}(x)h_{g^{-2}}(x)$, where
   $h_{g^{-1}}(x)$ and $h_{g^{-2}}(x)$ are the minimal polynomial of $g^{-1}$ and $g^{-2}$ over $\Bbb F_q$, respectively.

   For any $a, b \in \Bbb F_r$, the Hamming weight
of $c(a, b)$ is equal to $$W_H(c(a, b))= n-Z(r, a, b),$$ where
$$Z(r, a, b)=|\{x \in C_0^{(N, r)} : \mbox{Tr}_{r/q}(ax+bx^2)=0\}|.$$

Let $\phi$ be the canonical additive character of $\Bbb F_q$. Then $\psi=\phi \circ \mbox{Tr}_{r/q}$
is the canonical additive character of $\Bbb F_r$. By the orthogonal property of additive characters we have
\begin{eqnarray*} Z(r,a,b)&=&\sum_{x \in C_0^{(N, r)}}\frac 1 q \sum_{y \in \Bbb F_q}\phi(y \mbox{Tr}_{r/q}(ax+bx^2)) \\
&=& \sum_{x \in C_0^{(N, r)}}\frac 1 q \sum_{y \in \Bbb F_q}\phi \mbox{Tr}_{r/q}(yax+ybx^2)) \\
&=& \frac 1 q \cdot \frac{r-1} N+\frac 1 q \sum_{y \in \Bbb F_q^\ast} \sum_{x \in C_0^{(N, r)}}\psi(yax+ybx^2).
\end{eqnarray*}

\begin{thm}
 Let $r=q^2$ and $2N \mid (q+1)$. Then the weight distribution of the cyclic code $\mathcal C_2$ is give by Table 3
 if $\frac {q+1} {2N}$ is even and is given by Table 4 if $\frac {q+1} {2N}$ is odd.
\end{thm}

\[ \begin{tabular} {c} Table 3. Weight distribution of $\mathcal C_2$ when $\frac {q+1} {2N}$ is even \\
\begin{tabular}{|c|c|}
  \hline
 Weight & Frequency \\
  \hline
        0   &    1\\
   $\frac {(q-1)(r-1)} {q N}-\frac {q-1} q \eta_i^{(N, r)}$  & $\frac {r-1} N (0 \leq i \leq N-1)$ \\
   $\frac {(q-1)(r-1)} {q N}-\frac {2(q-1)} q \eta_j^{(2N, r)}$ &  $\frac {r-1} {2N} (0 \leq j \leq 2N-1)$  \\
  $\frac {(q-1)(r-1)} {qN}+\frac 2 q \eta_j^{(2N, r)} $ & $\frac {(N-1)(r-1)^2} {2N^2} (0 \leq j \leq 2N-1) $ \\
  $\frac {(q-1)(r-1)} {qN}-q+1+\frac 2 q \eta_0^{(2N, r)} $ & $\frac {(q-1)(r-1)} N$ \\
  $\frac {(q-1)(r-1)} {qN}+1+\frac 2 q \eta_0^{(2N, r)} $ & $\frac {r-1} N(\frac {r-1} {2N}-q+1)$ \\
  $\frac {(q-1)(r-1)} {qN}+1+\frac 2 q \eta_j^{(2N, r)} $ & $\frac {(r-1)^2} {2N^2} (1 \leq j \leq 2N-1)$  \\
  \hline
\end{tabular}
\end{tabular}
\]

\[ \begin{tabular} {c} Table 4. Weight distribution of $\mathcal C_2$ when $\frac {q+1} {2N}$ is odd \\
\begin{tabular}{|c|c|}
  \hline
 Weight & Frequency \\
  \hline
        0   &    1\\
   $\frac {(q-1)(r-1)} {q N}-\frac {q-1} q \eta_i^{(N, r)}$  & $\frac {r-1} N (0 \leq i \leq N-1)$ \\
   $\frac {(q-1)(r-1)} {q N}-\frac {2(q-1)} q \eta_j^{(2N, r)}$ &  $\frac {r-1} {2N} (0 \leq j \leq 2N-1)$  \\
  $\frac {(q-1)(r-1)} {qN}+\frac 2 q \eta_j^{(2N, r)} $ & $\frac {(N-1)(r-1)^2} {2N^2} (0 \leq j \leq 2N-1) $ \\
  $\frac {(q-1)(r-1)} {qN}-q+1+\frac 2 q \eta_N^{(2N, r)} $ & $\frac {(q-1)(r-1)} N$ \\
  $\frac {(q-1)(r-1)} {qN}+1+\frac 2 q \eta_N^{(2N, r)} $ & $\frac {r-1} N(\frac {r-1} {2N}-q+1)$ \\
  $\frac {(q-1)(r-1)} {qN}+1+\frac 2 q \eta_j^{(2N, r)} $ & $\frac {(r-1)^2} {2N^2} (0 \leq j \leq 2N-1, j \neq N)$  \\
  \hline
\end{tabular}
\end{tabular}
\]

\begin{proof}
Let $\alpha$ be a primitive element of $\Bbb F_r$ and $\beta=\alpha^{\frac {r-1}{q-1}}=\alpha^{q+1}$. Then
$\Bbb F_q^\ast =\langle \beta \rangle$ and $\Bbb F_q^\ast=C_0^{(2,q)} \cup \beta C_0^{(2, q)}$, where
$C_0^{(2, q)}=\langle \beta^2 \rangle$.  Since $q \equiv -1 \pmod {2N}$,
we have $\Bbb F_q^\ast \subset C_0^{(2N, r)} \subset C_0^{(N, r)}$
and $yC_0^{(N, r)}=C_0^{(N, r)}$ for each $y \in \Bbb F_q^\ast$.

 If  $a=0$ and $b=0$, then we have
  $$Z(r, a, b)=\frac {r-1} N.$$
  This value occurs once.

If $a \in C_i^{(N, r)}$ for some $i (0 \leq i \leq N-1)$ and $b=0$, then we have
 \begin{eqnarray*} Z(r,a,b)&=& \frac {r-1} {q N}+\frac 1 q \sum_{y \in \Bbb F_q^\ast} \sum_{x \in C_0^{(N, r)}}\psi(y ax) \\
&=& \frac {r-1} {q N}+\frac {q-1} q\sum_{x \in C_0^{(N, r)}}\psi(ax)
= \frac {r-1} {q N}+\frac {q-1} q\eta_i^{(N, r)}.
\end{eqnarray*}
This value occurs $\frac {r-1} N$ times.

 If $a=0$ and $b \neq 0$, then we can let $b \in C_j^{(2N, r)}$ for some $j, 0 \leq j \leq 2N-1$, by
  $2N \mid (r-1)$. We have
\begin{eqnarray*} Z(r,a,b)&=& \frac {r-1} {q N}+\frac 1 q \sum_{y \in \Bbb F_q^\ast} \sum_{x \in C_0^{(N, r)}}\psi(ybx^2)
= \frac {r-1} {q N}+\frac {q-1} q\sum_{x \in C_0^{(N, r)}}\psi(bx^2)\\
&=& \frac {r-1} {q N}+\frac {2(q-1)} q\sum_{x \in C_0^{(2N, r)}}\psi(bx)= \frac {r-1} {q N}+\frac {2(q-1)} q\eta_j^{(2N, r)}.
\end{eqnarray*}
 This value occurs $\frac {r-1} {2N}$ times.

 Now we suppose that $a \neq 0$ and $b \neq 0$. Then we have
\begin{eqnarray*} &&Z(r,a,b)= \frac {r-1} {q N}+\frac 1 q \sum_{y \in \Bbb F_q^\ast} \sum_{x \in C_0^{(N, r)}}\psi(yax+ybx^2) \\
&=& \frac {r-1} {q N}+\frac 1 q (\sum_{y \in C_0^{(2, q)}} \sum_{x \in C_0^{(N, r)}}\psi(yax+ybx^2)
+ \sum_{y \in C_1^{(2, q)}} \sum_{x \in C_0^{(N, r)}}\psi(yax+ybx^2)) \\
&=& \frac {r-1} {q N}+\frac 1 q (\sum_{y \in C_0^{(2, q)}} \sum_{x \in C_0^{(N, r)}}\psi(yax+ybx^2)
+ \sum_{y \in C_0^{(2, q)}} \sum_{x \in C_0^{(N, r)}}\psi(\beta y a x+ \beta ybx^2)) \\
&=& \frac {r-1} {q N}+\frac 1 q (\frac 1 2 \sum_{y \in \Bbb F_q^\ast} \sum_{x \in C_0^{(N, r)}}\psi(y^2ax+y^2bx^2)
+ \frac 1 2 \sum_{y \in \Bbb F_q^\ast} \sum_{x \in C_0^{(N, r)}}\psi(\beta y^2 a x+ \beta y^2bx^2)) \\
&=& \frac {r-1} {q N}+\frac 1 {2q} (\sum_{y \in \Bbb F_q^\ast} \sum_{x \in C_0^{(N, r)}}\psi(yaxy+b(x y)^2)
+  \sum_{y \in \Bbb F_q^\ast} \sum_{x \in C_0^{(N, r)}}\psi(\beta y a x y+ \beta b(x y)^2)) \\
&=& \frac {r-1} {q N}+\frac 1 {2q} (\sum_{y \in \Bbb F_q^\ast} \sum_{z \in C_0^{(N, r)}}\psi(yaz+bz^2)
+ \sum_{y \in \Bbb F_q^\ast} \sum_{z \in C_0^{(N, r)}}\psi(\beta y a z+ \beta bz^2)) \\
&=& \frac {r-1} {q N}+\frac 1 {2q} (\sum_{z \in C_0^{(N, r)}}\psi(bz^2)\sum_{y \in \Bbb F_q^\ast} \psi(y a z)
+ \sum_{z \in C_0^{(N, r)}}\psi(\beta bz^2)\sum_{y \in \Bbb F_q^\ast} \psi(\beta y a z) )\\
&=& \frac {r-1} {q N}+\frac 1 {2q} (\sum_{z \in C_0^{(N, r)}}\psi(bz^2) +\sum_{z \in C_0^{(N, r)}}\psi(\beta bz^2)) \sum_{y \in \Bbb F_q^\ast}\psi(y\mbox{Tr}_{r/q}( a z)).
\end{eqnarray*}
Note that $\beta=\alpha^{q+1}\in  \Bbb F_q^*$.
Suppose that $0\ne a\in C_0^{(N,r)} $ and $\mbox{Tr}_{r/q}( a z)=0$. Then we have $az+(az)^q=0$ and
 $$z=a^{-1}\alpha^{\frac {q+1}2} v \mbox{ for all } v \in \Bbb F_q^\ast.$$
 This means that there exist exactly $q-1$ solutions  $z \in C_0^{(N, r)}$ such that $\mbox{Tr}_{r/q}( a z)=0$ if $a \in C_0^{(N, r)}$ and
 there exists no solution  $z \in C_0^{(N, r)}$ such that $\mbox{Tr}_{r/q}( a z)=0$ if $a \not \in C_0^{(N, r)}$.

If $a \in C_i^{(N,r)}$ for all $i=1, 2, \ldots, N-1$ and $b \in C_j^{(2N, r)}$ for some $j$ $ (0 \leq j \leq 2N-1)$,
then we have  $\mbox{Tr}_{r/q}( a z)\neq 0$ and
\begin{eqnarray*} Z(r,a,b)&=& \frac {r-1} {q N}+\frac 1 {2q} (\sum_{z \in C_0^{(N, r)}}\psi(bz^2)+\sum_{z \in C_0^{(N, r)}}\psi(\beta bz^2))\cdot(-1)\\
&=& \frac {r-1} {q N}-\frac 2 q \sum_{z \in C_0^{(2N, r)}}\psi(bz)
=  \frac {r-1} {q N}-\frac 2 q\eta_j^{(2N, r)}.
\end{eqnarray*}
  This value occurs $\frac {r-1} N \cdot (N-1) \cdot \frac {r-1} {2N}=\frac {(N-1)(r-1)^2} {2N^2}$ times.

  In the following, we consider the case $a \in C_0^{(N,r)}$. We have
  \begin{eqnarray*} \Delta&= &(\sum_{z \in C_0^{(N, r)}}\psi(bz^2)+\sum_{z\in C_0^{(N,r)}}
  \psi(\beta bz^2))\sum_{y \in \Bbb F_q^\ast} \psi(y\mbox{Tr}_{r/q}( a z))\\ &=&(\sum_{\small\begin{array}{c} z \in C_0^{(N, r)} \\ \mbox{Tr}_{r/q}( a z)=0 \end{array}}
  \psi(bz^2) +\sum_{\small\begin{array}{c} z \in C_0^{(N, r)} \\ \mbox{Tr}_{r/q}( a z)=0 \end{array}}
  \psi(\beta bz^2)) \cdot (q-1)\\ &+& (\sum_{\small\begin{array}{c} z \in C_0^{(N, r)} \\ \mbox{Tr}_{r/q}( a z)\neq 0 \end{array}}\psi(bz^2)+\sum_{\small\begin{array}{c} z \in C_0^{(N, r)} \\ \mbox{Tr}_{r/q}( a z)\neq 0 \end{array}}\psi(\beta bz^2)) \cdot (-1)
 \\
  &=&(\sum_{\small\begin{array}{c} z \in C_0^{(N, r)} \\ \mbox{Tr}_{r/q}( a z)=0 \end{array}}
  \psi(bz^2) +\sum_{\small\begin{array}{c} z \in C_0^{(N, r)} \\ \mbox{Tr}_{r/q}( a z)=0 \end{array}}
  \psi(\beta bz^2))\cdot q \\&+& (\sum_{z \in C_0^{(N, r)}}\psi(bz^2) +\sum_{z \in C_0^{(N, r)}}\psi(\beta bz^2) )\cdot (-1)\\
  &=&q\sum_{v \in \Bbb F_q^\ast}
  \psi(b(a^{-1}\alpha^{\frac {q+1} 2}v)^2)+ q\sum_{v \in \Bbb F_q^\ast}
  \psi(\beta b(a^{-1}\alpha^{\frac {q+1} 2}v)^2)-2 \sum_{z \in C_0^{(N, r)}}\psi(bz^2) \\
  &=&q\sum_{v \in \Bbb F_q^\ast}
  \psi(ba^{-2}\alpha^{q+1}v^2)+ q\sum_{v \in \Bbb F_q^\ast}
  \psi(\beta ba^{-2}\alpha^{q+1}v^2)-4 \sum_{z \in C_0^{(2N, r)}}\psi(b z) \\
  &=&2q\sum_{v \in C_0^{(2, q)}}
  \psi(ba^{-2}v) + 2q\sum_{v \in C_0^{(2, q)}}
  \psi(ba^{-2}\beta v)-4 \sum_{z \in C_0^{(2N, r)}}\psi(b z) \\
  &=&2q\sum_{v \in \Bbb F_q^\ast}
  \psi(ba^{-2}v) -4 \sum_{z \in C_0^{(2N, r)}}\psi(b z)\\
  &=&2q\sum_{v \in \Bbb F_q^\ast}
  \psi(v\mbox{Tr}_{r/q}(ba^{-2})) -4 \sum_{z \in C_0^{(2N, r)}}\psi(b z).
\end{eqnarray*}

Suppose that $\mbox{Tr}_{r/q}(ba^{-2})=0$. Then we have $ba^{-2}+b^qa^{-2q}=0$ and
$$b=a^2 \alpha^{\frac {q+1} 2}v \mbox{ for all } v \in \Bbb F_q^\ast.$$
This means that there exist exactly $q-1$ solutions $b \in C_{\frac {q+1} 2}^{(2N, r)}$
such that $\mbox{Tr}_{r/q}(ba^{-2})=0$ for each $a \in C_0^{(N,r)}$,
where $C_{\frac {q+1} 2}^{(2N, r)}=C_{\frac {q+1} 2 \pmod {2N}}^{(2N, r)}$.

\begin{enumerate}
  \item If $\frac {q+1} {2N}$ is even, then $\frac {q+1} 2 \equiv 0 \pmod {2N}$.

   If $a \in C_0^{(N,r)}$ and $b \in C_0^{(2N, r)}$ satisfy $\mbox{Tr}_{r/q}(ba^{-2})=0$, then we have
 $$Z(r, a, b)=\frac {r-1} {q N}+q-1-\frac 2 q \eta_0^{(2N, r)}.$$
 This value occurs $\frac {(q-1)(r-1)} N$ times.

 If $a \in C_0^{(N,r)}, b \in C_0^{(2N, r)}$, and $\mbox{Tr}_{r/q}(ba^{-2})\neq 0$, then we have
 $$Z(r, a, b)=\frac {r-1} {q N}-1-\frac 2 q \eta_0^{(2N, r)}.$$
 This value occurs $\frac {r-1} N \cdot (\frac {r-1} {2N}-q+1)$ times.

 If $a \in C_0^{(N,r)}$ and $b \in C_j^{(2N, r)}$ for some $j$ $ (1 \leq j \leq 2N-1)$, then $\mbox{Tr}_{r/q}(ba^{-2})\neq 0$ and
 $$Z(r, a, b)=\frac {r-1} {q N}-1-\frac 2 q \eta_j^{(2N, r)}.$$
 This value occurs $\frac {r-1} N \cdot \frac {r-1} {2N}$ times.
  \item  If $\frac {q+1} {2N}$ is odd, then $\frac {q+1} 2 \equiv N \pmod {2N}$.

   If $a \in C_0^{(N,r)}$ and $b \in C_N^{(2N, r)}$ satisfy $\mbox{Tr}_{r/q}(ba^{-2})=0$, then we have
 $$Z(r, a, b)=\frac {r-1} {q N}+q-1-\frac 2 q \eta_N^{(2N, r)}.$$
 This value occurs $\frac {(q-1)(r-1)} N$ times.

 If $a \in C_0^{(N,r)}, b \in C_N^{(2N, r)}$, and $\mbox{Tr}_{r/q}(ba^{-2})\neq 0$, then we have
 $$Z(r, a, b)=\frac {r-1} {q N}-1-\frac 2 q \eta_N^{(2N, r)}.$$
 This value occurs $\frac {r-1} N \cdot (\frac {r-1} {2N}-q+1)$ times.

 If $a \in C_0^{(N,r)}$ and $b \in C_j^{(2N, r)}$ for some $j$ $(0 \leq j \leq 2N-1, j \neq N)$, then $\mbox{Tr}_{r/q}(ba^{-2})\neq 0$ and
 $$Z(r, a, b)=\frac {r-1} {q N}-1-\frac 2 q \eta_j^{(2N, r)}.$$
 This value occurs $\frac {r-1} N \cdot \frac {r-1} {2N}$ times.
\end{enumerate}

 Note that $W_H(c(a, b))= n-Z(r, a, b)$. Then we can obtain Table 3 and Table 4. This completes the proof.
 \end{proof}

We have determined the weight distribution of the cyclic code $\mathcal C_2$ when Gauss periods of order $2N$ are known.
Then we have the following theorem.

\begin{thm}
Assume that there exists the least positive integer $e$ such that
$p^e \equiv -1\pmod {2N}$, we know that $p$ is odd. Let $q=p^{ef}$ for some positive integer $f$ and $r=q^2=p^{2ef}$.
\begin{enumerate}
  \item If $f$ and $\frac {p^e+1} {2N}$ are both odd, then
  the weight distribution of $\mathcal C_2$ can be given by Table 5.
  \item In all other cases,
  the weight distribution of $\mathcal C_2$ can be given by Table 6.
\end{enumerate}
\end{thm}

\[ \begin{tabular} {c} Table 5. The case (1) of Theorem 4.2. \\
\begin{tabular}{|c|c|}
  \hline
 Weight & Frequency \\
  \hline
        0   &    1\\
   $\frac {r-1} N-q+1$  & $\frac {r-1} N$ \\
   $\frac {r-1} N$ &  $\frac {(4N-3)(r-1)} {2N}$  \\
  $\frac {r-1} N-2q+2$ & $\frac {r-1} {2N} $ \\
  $\frac {(q+1)(q-2)} N+2$ & $\frac {(N-1)(r-1)^2} {2N^2}$ \\
  $\frac {(q+1)(q-2)} N$   &  $\frac {(N-1)(2N-1)(r-1)^2} {2N^2}$ \\
  $\frac {(q+1)(q-2)} N-q+3$   &  $\frac {(q-1)(r-1)} N$ \\
  $\frac {(q+1)(q-2)} N+3$ & $\frac {r-1} N(\frac {r-1} {2N}-q+1)$ \\
  $\frac {(q+1)(q-2)} N+1$ & $\frac {(2N-1)(r-1)^2} {2N^2}$  \\
  \hline
\end{tabular}
\end{tabular}
\]

\[ \begin{tabular} {c} Table 6. The case (2) of Theorem 4.2. \\
\begin{tabular}{|c|c|}
  \hline
 Weight & Frequency \\
  \hline
        0   &    1\\
   $\frac {(q-1)(q+(-1)^f(N-1))} N$  & $\frac {r-1} N$ \\
   $\frac {(q-1)(q-(-1)^f)} N$ &  $\frac {(4N-3)(r-1)} {2N}$  \\
  $\frac {(q-1)(q+(-1)^f(2N-1))} N$ & $\frac {r-1} {2N} $ \\
  $\frac {q^2-q-1+(-1)^f} N-2\cdot (-1)^f$ & $\frac {(N-1)(r-1)^2} {2N^2}$ \\
  $\frac {q^2-q-1+(-1)^f} N$   &  $\frac {(N-1)(2N-1)(r-1)^2} {2N^2}$ \\
  $\frac {q^2-q-1+(-1)^f} N-2\cdot (-1)^f-q+1$   &  $\frac {(q-1)(r-1)} N$ \\
  $\frac {q^2-q-1+(-1)^f} N-2\cdot (-1)^f+1$ & $\frac {r-1} N(\frac {r-1} {2N}-q+1)$ \\
  $\frac {q^2-q-1+(-1)^f} N+1$ & $\frac {(2N-1)(r-1)^2} {2N^2}$  \\
  \hline
\end{tabular}
\end{tabular}
\]

\begin{proof}
  (1) If $f, p$, and $\frac {p^e+1} {2N}$ are all odd, then by Lemma 2.4 we have
  $$\eta_N^{(2N, r)}=\frac {(2N-1)q-1} {2N}, \eta_j^{(2N, r)}=-\frac {q+1}{2N} \mbox{ for } 0 \leq j \leq 2N-1, j \neq N.$$
  Thus
  $$\eta_0^{(N,r)}=\eta_0^{(2N,r)}+\eta_N^{(2N,r)}=\frac {(N-1)q-1} N,$$
  $$\eta_i^{(N,r)}=\eta_i^{(2N,r)}+\eta_{i+N}^{(2N,r)}=-\frac {q+1} N \mbox{ for } 1 \leq i \leq N-1.$$
  Note that $$\frac {q+1} {2N}=\frac {p^{ef}+1} {2N}=\frac {p^e+1} {2N} \cdot (p^{e(f-1)}+(-1)^{f-2}p^{e(f-2)}+\cdots+(-1)p^e+1)$$ is odd
  since $f, p$, and $\frac {p^e+1} {2N}$ are all odd.
  Then Table 5 can be obtained by Table 4 of Theorem 4.1.

  (2) In all other cases, by Lemma 2.4 we have
  $$\eta_0^{(2N, r)}=\frac {(-1)^{f+1}(2N-1)q-1} {2N}, \eta_j^{(2N, r)}=\frac {(-1)^fq-1}{2N} \mbox{ for } 1 \leq j \leq 2N-1.$$
  Thus
  $$\eta_0^{(N,r)}=\eta_0^{(2N,r)}+\eta_N^{(2N,r)}=\frac {(-1)^{f+1}(N-1)q-1} N,$$
  $$\eta_i^{(N,r)}=\eta_i^{(2N,r)}+\eta_{i+N}^{(2N,r)}=\frac {(-1)^fq-1} N \mbox{ for } 1 \leq i \leq N-1.$$
  Note that there is at least one of  $f$ and $\frac {p^e+1} {2N}$ is even and $p$ is odd.
  Then $$\frac {q+1} {2N}=\frac {p^{ef}+1} {2N}=\frac {p^e+1} {2N} \cdot (p^{e(f-1)}+(-1)^{f-2}p^{e(f-2)}+\cdots+(-1)p^e+1)$$ is even and
  Table 6 can be obtained by Table 3 of Theorem 4.1.
\end{proof}

\begin{exa}
Let $q=p=3, r=9$, and $N=2$. By Table 5 we known that $\mathcal C_2$ is a $[4, 3, 2]$ cyclic code over $\Bbb F_3$ and thus is optimal
with respect to Singleton bound. The weight enumerator of such cyclic code is
$$1+12x^2+8x^3+6x^4.$$
\end{exa}

\begin{exa} Let $q=p=7$, $r=49$, and $N=2$. By Table 6 we known that $\mathcal C_2$ is a $[24, 4, 12]$ cyclic code over $\Bbb F_7$
and the weight enumerator is
$$1+12x^{12}+144x^{16}+24x^{18}+864x^{20}+864x^{21}+288x^{22}+144x^{23}+60x^{24}.$$
 \end{exa}

 \begin{exa} Let $q=p=5$, $r=25$, and $N=3$. By Table 5 we known that $\mathcal C_2$ is an $[8, 3, 4]$ cyclic code over $\Bbb F_5$
and the weight enumerator is
$$1+8x^4+64x^6+32x^7+20x^8.$$
 \end{exa}

We also have a tight bound on the minimum weight of $\mathcal C_2$ which is denote $W_H(\mathcal C_2)$.

\begin{thm}
 Let $r=q^2$, $2N \mid (q+1)$, and $5 \leq N < \frac {\sqrt r} 2$. Then we have
 $$W_H(\mathcal C_2) \geq \frac {(q-1)(r-(2N-1) \sqrt r)} {qN}.$$
\end{thm}

\begin{proof}
By Lemma 2.6 and Theorem 4.1 we only need to compare some values of the weights.
It is not difficult and then we omit
the proof here.
\end{proof}

\end{document}